\documentclass{article}
\usepackage{amsfonts}
\usepackage{amsmath}
\usepackage{fullpage}
\usepackage{graphicx}
\usepackage{subfigure}
\usepackage{color,hyperref,cite}
\usepackage{float}
\usepackage{algorithm}
\usepackage[noend]{algorithmic}
\usepackage[capitalise]{cleveref}

\usepackage{amsmath,amsthm,amssymb,enumerate,amsfonts,cite}
\usepackage{algorithm}
\usepackage[noend]{algorithmic}
\usepackage{xspace}
\usepackage{nicefrac}

\newcommand{\mcal}{\ensuremath{\mathcal{M}}\xspace}
\newcommand{\scal}{\ensuremath{\mathcal{S}}\xspace}

\newcommand{\acal}{\ensuremath{\mathcal{A}}\xspace}
\newcommand{\gcal}{\ensuremath{\mathcal{G}}\xspace}
\newcommand{\ical}{\ensuremath{\mathcal{I}}\xspace}
\newcommand{\pcal}{\ensuremath{\mathcal{P}}\xspace}

\newcommand{\argmax}{\ensuremath{\text{argmax}}\xspace}
\newcommand{\pr}{\ensuremath{\text{Pr}}\xspace}
\newcommand{\opt}{\ensuremath{\text{OPT}}\xspace}

\newcommand{\ebb}{\ensuremath{\mathbb{E}}\xspace}

\newtheorem{theorem}{Theorem}

\newtheorem{definition}{Definition}

\newtheorem{lemma}{Lemma}

\begin{document}

\title{Matching in Stochastically Evolving Graphs\thanks{Supported by 
the NeST initiative of the EEE/CS School of the University of Liverpool and by the EPSRC grants EP/P020372/1 and EP/P02002X/1.}}
\author{Eleni C.~Akrida\thanks{Department of Computer Science, Durham University, UK.
Email: \texttt{eleni.akrida@durham.ac.uk}} \and 
Argyrios Deligkas\thanks{Department of Computer Science, Royal Holloway University of London, UK. 
Email: \texttt{argyrios.deligkas@rhul.ac.uk}} \and 
George B. Mertzios\thanks{Department of Computer Science, Durham University, UK.
Email: \texttt{george.mertzios@durham.ac.uk}} \and 
Paul G. Spirakis\thanks{Department of Computer Science, University of Liverpool, UK, and 
Computer Engineering \& Informatics Department, University of Patras, Greece. 
Email: \texttt{p.spirakis@liverpool.ac.uk}} \and 
Viktor Zamaraev\thanks{Department of Computer Science, University of Liverpool, Liverpool, UK. 
Email: \texttt{viktor.zamaraev@liverpool.ac.uk}}}

\date{\vspace{-0.5cm}}
\maketitle

\begin{abstract}
This paper studies the maximum cardinality matching  problem in stochastically evolving graphs.
We formally define the arrival-departure model with stochastic departures. 
There, a graph is sampled from a specific probability distribution and it is revealed as a series of 
snapshots. 
Our goal is to study algorithms that create a large matching in the sampled graphs.
We define the price of stochasticity for this problem which intuitively captures the loss of {\em any} 
algorithm in the worst case in the size of the matching due to the uncertainty of the model. 
Furthermore, we prove the existence of a {\em deterministic} optimal algorithm for the problem.
In our second set of results we show that we can efficiently approximate the expected size of a maximum 
cardinality matching by deriving a fully randomized approximation scheme (FPRAS) for it. The FPRAS is the 
backbone of a probabilistic algorithm that is {\em optimal} when the model is defined over two timesteps. 
Our last result is an upper bound of $\frac{2}{3}$ on the price of stochasticity. This means that there is no 
algorithm that can match more than $\frac{2}{3}$ of the edges of an optimal matching in hindsight.\newline

\end{abstract}

\section{Introduction}
\label{sec:intro}
Matching is one of the most fundamental problems in Algorithms, receiving a lot of attention recently due 
to its natural applications in several fields, from medicine to economics, biology and computer science.
Examples include market clearing where the goal is to assign as many items of a list of available goods to 
interested buyers, as well as kidney exchange where the goal is to match as many patients and compatible 
donors as possible. 
Many of the above application domains require inherently dynamic models to capture the arrival and departure 
of people, goods, or amenities over time with the goal remaining to match as may pairs as possible.

In this work, we propose a stochastic, discrete-time dynamic graph model in which vertices are born and  stochastically
die over time, and the objective is to find maximum cardinality matchings. 
An instance of the problem is a graph $G = (V,E)$ in which every vertex $v\in V$ arrives in (the morning of) some known day $a_v\in \mathbb{N}$ and will be alive in the graph until (the night of) some day $d_v$ which is a random variable with a known discrete probability distribution on the sample space $\{a_v,a_v+1,\ldots, b_v\}$, where $b_v\in \mathbb{N}, b_v\geq a_v$ is also known. We call $b_v$ the \emph{deadline} of $v$ and $d_v$ its \emph{(actual) death time}. 
A vertex $v$ may become connected via edges only to other vertices that are alive during $v$'s lifetime; those edges exist in the graph only at days of existence of both endpoints. The objective of a maximum cardinality matching translates here into finding as many pairs of vertices that become connected via an ``alive'' edge (at some point in time) and matching those pairs so that no two vertices are matched to the same vertex.

To further motivate our model, consider adverts on YouTube. Youtube has a number of adverts to serve its viewers every day by presenting an ad within one of its videos; one can present an ad at any point until the end of the video, but they do not know exactly when the viewer will change between videos.
Further applications can be found in the dating/matchmaking apps market: suppose a group of people who do not know each other but all use a particular dating app plan a trip to Spain; each of them knows when his/her flight lands in, say, Barcelona and when he/she will go back to their hometown, but with some probability they may leave Barcelona to visit other cities nearby. If they could notify their dating app about their plans, how can the app suggest as many couples' matchings as possible for the duration of their stay in Barcelona?

Any algorithm solving this problem needs to be adaptive in nature, in the sense that it receives the initial information as its input but also learns the evolution of the graph over time and thus may adapt to the new information: we know in advance both the arrival time $a_v$ and the deadline $b_v$ of every vertex $v \in V$ (namely, in the above example, when everyone's flight lands and departs from Barcelona), but the actual death time $d_v$ of $v$ is only revealed to us after the death takes effect (namely, in the above example, we find out if someone took the train from Barcelona to visit, e.g.~Madrid, only after they board the train). Also, although the underlying graph $G$ is known in advance, in general the actual set of edges that become incident to a vertex $v$ during its lifetime can only be known after $v$'s death time. The fact that our algorithms do not know the exact death time of a vertex until the day after it dies is a main difference between our model and previous studies on stochastic and online matchings (see related work in Section~\ref{sec:related_work}).
An (adaptive) algorithm, therefore, has to make decisions adaptively as well; that is, the algorithm may make a \emph{tentative} matching of an alive vertex $v$ to some alive neighbor $u$ (if any), but that decision is subject to change up until the day that the algorithm decides to actually match $v$. Matches (non-tentative ones) once made cannot be revoked. 

Figure~\ref{fig:pi_graph} shows an example of a graph with $4$ vertices, for each of which we know when they arrive and when they will depart \emph{at the latest}. Suppose also that each vertex will die at some point during its [arrival,deadline] interval \emph{uniformly at random} and independently of other vertices. What is the best set of edges that an adaptive algorithm can select to be added in the matching given that information? 

\begin{figure}[h]
	\centering
	\includegraphics[width=.2\textwidth]{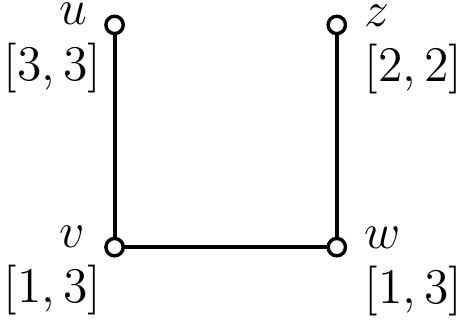}
	\caption{Example of a graph with vertex arrival and departure times.}\label{fig:pi_graph}
\end{figure}

To answer this, we define a \emph{realization} of the stochastic graph (over time) which is a particular evolution of the graph in time, i.e.~a death time per vertex (chosen according to the stochastic model). One can find a maximum matching by viewing this realization as a static graph and computing the maximum matching using known polynomial-time algorithms~\cite{micali_vazirani_matching}. As an example, assume that in the graph of Figure~\ref{fig:pi_graph} both $v$ and $w$ die at time $t=1$; recall that since death times occur at the very end of the day, $v$ and $w$ are connected via the edge $vw$ in day $1$ in this realization, and therefore can be matched. The realization occurs with probability $\nicefrac{1}{9}$ and can be indeed viewed as the static graph containing the single edge $vw$.
Notice that the edge $vw$ is present in any realization of the graph in Figure~\ref{fig:pi_graph}, while $uv$ is present with probability $\nicefrac{1}{3}$ and $wz$ with probability $\nicefrac{2}{3}$.

\subsection*{Our results}
Our results are threefold. Firstly, we formally define the maximum cardinality matching problem  in the  
arrival-departure model with stochastic departures. In addition, we define the price of stochasticity for this problem 
which intuitively captures  the loss of {\em any} algorithm in the worst case in the size of the matching due to the 
uncertainty of the model. Furthermore, we prove that there exist a {\em deterministic} optimal algorithm for the problem.
Second, we show that we can efficiently approximate the expected size of a maximum cardinality matching by deriving  
a fully randomized approximation scheme (FPRAS) for it. The FPRAS is the backbone of a probabilistic algorithm that is
{\em optimal} when the model is defined over two timesteps. Our last result is an upper bound of $\frac{2}{3}$ on the 
price of stochasticity. This means that there is no algorithm that can match more than $\frac{2}{3}$ of the edges of 
an optimal matching in hindsight.

\subsection{Related work}
\label{sec:related_work}
The problem of finding a maximum matching in a graph, i.e.~ a maximum-cardinality set of edges without common vertices, has been studied from a static point of view for many years with different variations regarding the class of graphs considered, or whether the edges are weighted or not; in the former case, the objective is to find matchings of maximum total weight.

Many matching processes, however, are inherently dynamic with participants arriving and matches being created over time. 
Such is the case also in online matchings with relevant literature being relatively recent and focused on online algorithms.
Online Bipartite Matching is the problem where a bipartite graph's left-hand-side is known in advance, while vertices
on the right-hand-side arrive online in an arbitrary order; on the arrival of a vertex, its incident
edges are revealed and the algorithm must irrevocably either match it to one of its unmatched
neighbors or leave it unmatched. Karp et al.~\cite{kvv90} introduced the Ranking algorithm and proved that it is the best possible among online algorithms. Its analysis has since been simplified (see, e.g. Devanur et al.~\cite{DevanurJK13} whose approach also extends to online vertex-weighted integral matching).
Huang et al.~\cite{KTWZZ18} study maximum cardinality matching in a fully online model where \emph{all} vertices arrive online, the incident edges (to previously-arrived vertices) as well as a fixed death time (known to the algorithm) for each vertex is revealed on arrival.
Ashlagi et al.~\cite{ashlagi2018} study the problem of (weighted) matching of agents who arrive at a marketplace over time and leave after $d$ time periods. They provide a $\nicefrac{1}{4}$-competitive algorithm over any sequence of arrivals when there is no a priori information about the weights or arrival times, and show that no algorithm is $\nicefrac{1}{2}$-competitive.
The problem of online market clearing where there is one commodity in the market being bought and sold by multiple buyers and sellers whose bids arrive and expire at different times is studied in~\cite{jacmSZ06}.
Lee and Singla~\cite{LeeS17} give the first positive results on an online matching problem, where edges are revealed in two stages; in each stage one has to immediately and irrevocably extend their matching using the edges from that stage.

Unlike all above-mentioned models, the vertex arrivals in our model are known in advance. However, the death time of a vertex is not fixed/deterministic but is instead a random variable with a discrete probability distribution on the sample space of discrete time steps from its arrival to its deadline. 
Bansal et al.~\cite{BansalGupta12} consider a different stochastic matchings problem: a random graph where each possible edge is present independently with some probability is given, and the goal is to build a large/heavy matching in the randomly generated graph given those probabilities. Unlike our model, they can only find out if an edge is present by querying it, and if it is indeed present in the graph, then they are forced to add it to their matching; their goal is to adaptively query the edges to maximize the expected weight of the matching.

The above literature, as well as this paper, examines inherently dynamic settings for the purpose of finding maximum matchings. The area of dynamic networks in general has flourished in recent years, and the notion of dynamic/temporal graphs is not new. Due to their vast applicability in many areas, temporal graph models have been studied from different perspectives under various names such as time-varying~\cite{krizanc1,flocchiniMS09,TangMML10-ACM}, evolving~\cite{xuan,clementiMMPS10,Ferreira-MANETS-04}, dynamic~\cite{GiakkoupisSS14}, temporal~\cite{AkridaMSZ-JCSS20,akridaTOCS,michailTSP},and graphs over time~\cite{Leskovec-Kleinberg-Faloutsos07}.
Notably, dynamic graphs that evolve stochastically have been studied before, e.g.~for the purpose of determining the speed of information spreading~\cite{clementi_rumor_spreading,clementiMMPS10,AkridaGMS_JPDC}.
For a recent attempt to integrate existing models, concepts, and results see the survey papers~\cite{flocchini1,flocchini2,CasteigtsFloccini12,michailCACM} and the references therein.

\section{Preliminaries}
\label{sec:prelims}

Let $G=(V,E)$ be a simple, undirected, unweighted graph. For every $S \subseteq V$,
let $E(S):= \{uv \in E: u \in S ~\text{and}~ v \in S \}$, i.e., $E(S)$ denote the edges 
of $G$ induced by $S$. Furthermore, for every $X \subseteq E$, let
$V(X):=\{v \in V: \exists u \in V~ \text{such that}~ uv \in X \}$.
Throughout the paper, we assume that every graph $G=(V,E)$ is associated with an 
{\em arrival } function $a: V \rightarrow \mathbb{N}$, and a {\em departure}
function $b: V \rightarrow \mathbb{N}$. 
For every $v \in V$ it holds $a(v) = a_v \leq b(v)=b_v$ and we call 
$[a_v,b_v]$ the lifetime of $v$.
We use $\langle G, a, b \rangle$ to denote this association and we term 
$G$ as the {\em underlying graph}. 
Furthermore, we denote $T = \max_{v \in V}b_v$ and we term it as the 
{\em lifetime} of $G$.

\begin{definition}[Arrival-departure graph]\label{def:ad-graph}
$\langle G=(V,E), a, b \rangle$ is an {\em arrival-departure} graph if for
every $uv \in E$ it holds that $[a_v,b_v] \cap [a_u,b_u] \neq \emptyset$, 
i.e., two vertices can be adjacent if their life intervals intersect.
The {\em realization} of an arrival-departure graph 
$\langle G=(V,E), a, b \rangle$ is a sequence of $T$ induced subgraphs of 
$G$ where the $i$-th subgraph is defined by the set 
$V_i = \{v \in V: a_v \leq i \leq b_v\}$.
\end{definition}

So, an arrival-departure graph is the description of a {\em dynamic} graph
whose set of edges {\em changes over time}. In arrival-departure graphs,
at any time $t \leq T$ an algorithm can perform operations {\em only} on 
the part of the graph that is available at this time. 
Hence, in an arrival-departure graph $\langle G=(V,E), a, b \rangle$, at any
time $t$ any algorithm can operate only on the induced subgraph defined by
$V_t$, which we term {\em snapshot} of $G$ at time $t$.

We complete the definition of arrival-departure graphs by ``endowing'' each 
vertex $v \in V$ with a probability distribution $P_v$ defined on 
$[a_v,b_v]$. Every $P_v$ independent from the probability distributions 
of the rest of the vertices of $V$. $P_v$ defines the {\em death} of vertex $v$. 
If vertex $v$ dies at time $d_v$, then it disappears and thus it does not
belong to any induced subgraph of $G$ after time $d_v$. The crucial point 
in our generalization, is that the death time of each vertex is {\em not}
known in advance but it is revealed {\em only after} it happens.
We formalize the above in the following definition. 

\begin{definition}[Stochastic Arrival-Departure Model]\label{def:stn}
Let $\langle G,a,b\rangle$ be an arrival-departure graph, where $G=(V,E)$,
and let $\pcal = \{ P_v ~|~ v \in V \}$ be a family of independent 
discrete probability distributions, where $P_v$ is defined on $[a_v, b_v]$.
The {\em stochastic arrival-departure model} $\gcal := \langle G,a,b,\pcal \rangle$ 
is the probability space over the set of all possible arrival-departure graphs 
$\langle G, a, b \rangle$ defined by setting $d_v$ according to $P_v$ for every  $v \in V$. 
\end{definition}

An {\em instantiation} of $\gcal$  is a graph $G'$, which is 
a subgraph of $G$, that it is revealed by a sequence of at most $T$ snapshots.
Note that an instantiation is {\em not} the same as a realization. An instantiation of \gcal
is a static graph while a realization is a sequence of subgraphs. Observe that an instantiation
can be produced by more than one realizations.
We denote $\ical_\gcal$ the set of possible instantiations of \gcal and $\pr(I)$ 
the probability that instantiation $I \in \ical_\gcal$ is realized.
Observe, given any stochastic arrival-departure model $\gcal = \langle G,a,b,\pcal \rangle$ 
at any time $t$, any snapshot $s$ at time $t$ of any instantiation of \gcal 
uniquely defines a stochastic arrival-departure model $\gcal'=\gcal(s,t)$; where,
by overloading notation, $s$ denotes the alive vertices of the instantiation at timestep $t$.
Formally, $\gcal(s,t)=\langle G', a', b, \pcal' \rangle$ where $G'=(V',E')$ and
\begin{itemize}
\item $V' :=\{v \in V: a_v \geq t~\text{or}~v \in s\}$;
\item $E' :=\{uv \in E: u \in V'~\text{and}~v \in V' \}$;
\item $a'_v = a_v$ if $v \in V'-s$, and $a'_v = t$ if $v \in s$;
\item $P'_v = P_v$ if $v \in V'-s$, and $P'_v$ is equal to $P_v$ conditioned on the
fact that $v$ is alive until time $t$ if $v \in s$.
\end{itemize}

We assume that every vertex arrives after time 1 and that at time 0 no vertices of $G=(V,E)$ 
are realized. For notation simplicity, we write $\gcal = \gcal(\emptyset, 0)$. Furthermore, 
$\gcal(s,t)=\emptyset$ for every $t>T$.

\begin{figure}[h]
    \centering
    \includegraphics[width=1.0\textwidth]{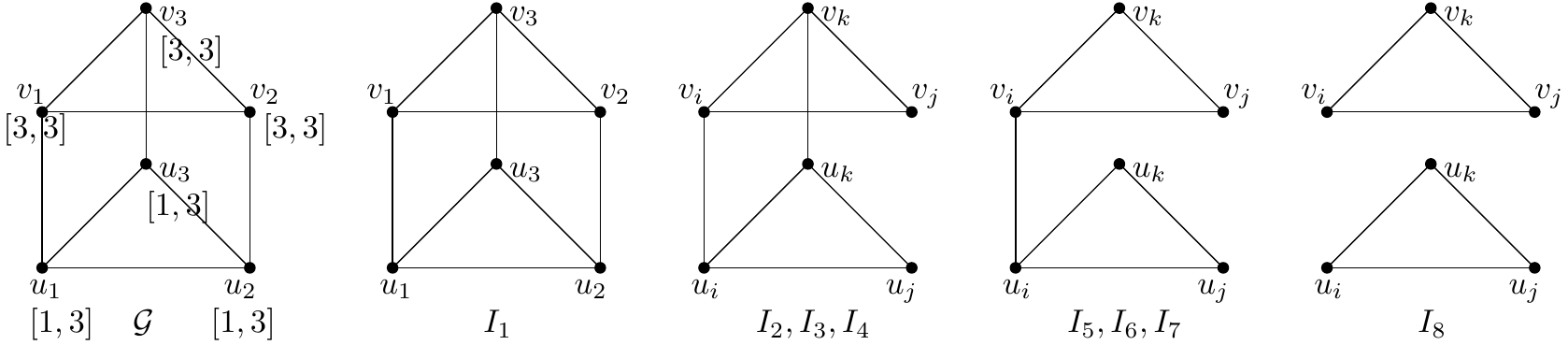}
    \caption{A stochastic arrival-departure model $\gcal$ and its corresponding instantiations. 
    \gcal consists of the vertices $u_1, u_2, u_3$, which arrive at timestep 1 and  departure 
    until timestep 3, and the vertices $v_1, v_2, v_3$, which arrive and departure at timestep 3.
    All vertices $u_i$ have the same probability distribution on their departure time 
    $P_{u_i} = (\varepsilon, \varepsilon, 1-2\varepsilon)$. Thus each vertex $u_i$ with probability 
    $\varepsilon$ departs at timestep 1, $\varepsilon$ departs at timestep 2, and with probability 
    $1-2\varepsilon$ at timestep 3. 
    \gcal has eight possible instantiations $I_1, I_2, \ldots, I_8$ depicted next to \gcal, where
    it holds $\Pr(I_1) = (1-\varepsilon)^3$ and $\Pr(I_2) = \Pr(I_3) =\Pr(I_4) = (1-\varepsilon)^2 \cdot 6 \varepsilon$
    and $\Pr(I_5) = \Pr(I_6) =\Pr(I_7) = (1-\varepsilon)\cdot 12\varepsilon^2$ and $\Pr(I_8) = 8 \varepsilon^3$.}
    \label{fig:STG}
\end{figure}

\subsection{Matching in the Stochastic Arrival-Departure Model}
\label{sec:matching-model}
We study the maximum matching problem in the stochastic arrival-departure model. 
Recall, a matching in a graph is a collection of independent edges.
As already mentioned, given an arrival-departure graph, at any timestep
$t$ any algorithm can operate only on the snapshot $s$ of the instantiation of \gcal. 
A realization of an instantiation from \gcal is revealed to any algorithm as 
follows. 
At time 0, the death time for every vertex $v$ is independently and randomly chosen
according to $P_v$. These death times are {\em unknown} to the algorithm and they 
define an instantiation of \gcal which is revealed to the algorithm
as a sequence of snapshots.
Independently, at each timestep $t$ the algorithm decides which edges available at snapshot 
$s$ to match irrevocably, {\em without} knowing which vertices will be dead at time 
$t+1$. After the algorithm matches a set $M \subseteq E(s)$, the matched vertices, 
$V(M)$, are removed from $s$.
Then, the remaining vertices of $s$ whose death time is $t$ are removed from $s$ 
and the time proceeds to time $t+1$. The new snapshot $s'$ of timestep $t+1$ contains 
all the unmatched vertices of $s$ that have remained alive and the vertices of $G$ that 
arrive at timestep $t+1$.

 We formalize the above mentioned by describing a general adaptive framework 
that captures any matching algorithm in the stochastic arrival-departure model.

\begin{algorithm}[htb]
\caption{General Adaptive Matching Algorithm} \label{alg:gaa}
\begin{algorithmic}[1]
 \REQUIRE{A stochastic arrival-departure model $\gcal =\langle G,a,b, \pcal \rangle$.}
 \ENSURE{A matching $M$ on the instantiation of \gcal.}
\STATE{$s \leftarrow \emptyset$; $M \leftarrow \emptyset$;}
\STATE{Every vertex $v \in V$ randomly and independently samples its death time according to $P_v$;}
\STATE{Let $D_t$ be the set of vertices with death time $t \in [1,T]$;}
\FOR{time step $t \in [1, T]$}
\STATE{$s \leftarrow s \cup S_t$, where $S_t = \{v \in V: a_v = t \}$;}
\STATE{Decide which edges $M_t \subseteq E(s)$ to match irrevocably;} \label{gaa:step-choice}
\STATE{$M \leftarrow M \cup M_t$;}
\STATE{$s \leftarrow s-\{V(M_t) \cup D_t\}$;}
\ENDFOR
\RETURN{Matching $M$;}
\end{algorithmic}
\end{algorithm}

Let \acal be the set of adaptive algorithms that work as described 
in Algorithm~\ref{alg:gaa}.
Fix any adaptive algorithm $A \in \acal$;  $A$ can be deterministic 
or randomized, depending on how it chooses which edges $M_t \subseteq E(s)$ 
to match irrevocably at Step~\ref{gaa:step-choice}.
For every $I \in \ical_\gcal$ we use $A(I)$ to denote the (expected) size 
of the matching $A$ produces on instantiation $I$. 
The {\em performance} of algorithm $A$ on \gcal is defined as 
\begin{align*}
\chi_A(\gcal) := \sum_{I \in \ical_\gcal} \pr(I) \cdot A(I).
\end{align*}

\subsection{Optimal matchings Vs Optimal algorithms}
An optimal matching for $I \in \ical_\gcal$ is a maximum matching of 
$I$ at hindsight; denoted $\opt(I)$. Thus, the expected size of optimal 
matching in \gcal is 
\begin{align*}
\opt(\gcal) := \sum_{I \in \ical_\gcal} \pr(I) \cdot \opt(I).
\end{align*}

In many cases $\opt(\gcal)$ cannot be obtained by any adaptive algorithm. This is because 
of the uncertainty of the realized graph. For this reason, we 
define the {\em optimal performance} of an algorithm for a given \gcal as 
\begin{align*}
\chi^*(\gcal) =  \max_{A \in \acal} \chi_A(\gcal)
\end{align*}
and an algorithm $A \in \acal$ is {\em optimal} for \gcal if 
$\chi^*(\gcal) =  \chi_A(\gcal)$.

Finally, we define the {\em stochasticity ratio} that captures the inefficiency of 
the optimal algorithm due to the stochastic nature of the model as:
\begin{align*}
\phi^* := \min_\gcal \frac{\chi^*(\gcal)}{\opt(\gcal)}.
\end{align*}

Let us demonstrate the notions discussed above on \gcal from Figure~\ref{fig:STG}.
So, we have that $\opt(I_1)=3$ and $\opt(I_i)=2$ for $i= 2, \ldots, 8$. If $\varepsilon \to 0$, then 
it is not hard to verify that the optimal algorithm proceeds as follows. 
\begin{itemize}
\item At timestep 1: no edges  are matched.
\item At timestep 2: if some vertices have died, then it matches any available edge; else it does not 
match any edges.
\item At timestep 3: It chooses a maximum matching for the snapshot of \gcal.
\end{itemize}
The example above shows us for any optimal algorithm it does not suffice to match edges {\em only} when
some new vertices arrive, but it has to consider matching available edges whenever there is a death, or
an arrival of a vertex; see Figure~\ref{fig:STG-counterexample}.
\begin{figure}[h]
    \centering
    \includegraphics[width=0.8\textwidth]{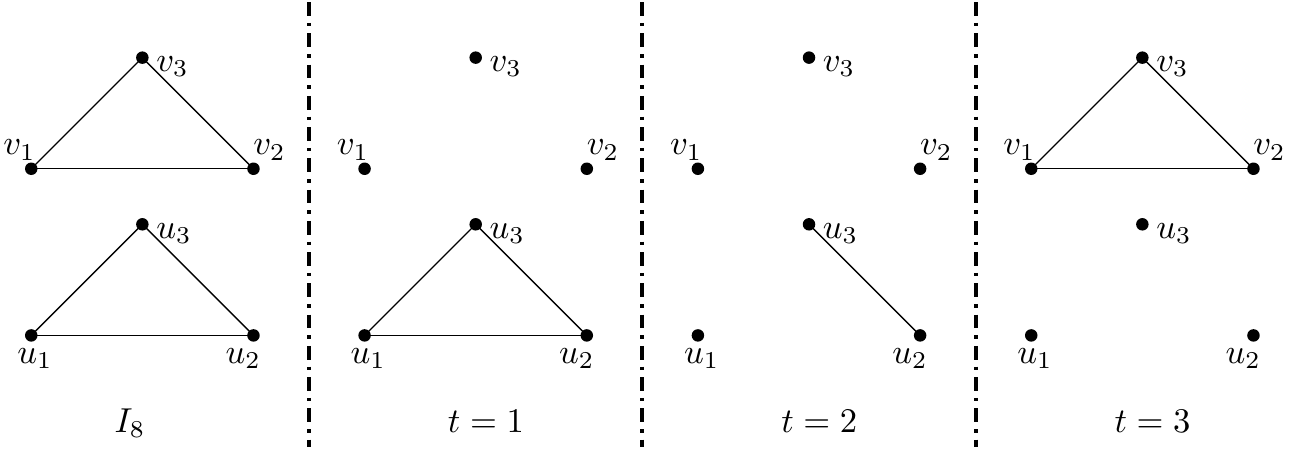}
    \caption{A realization of instantiation $I_8$ from the example of Figure~\ref{fig:STG} where $u_1$ dies after
    $t=1$ and $u_2$ and $u_3$ die at $t=2$. If an algorithm $A$ does not match the edge $u_2u_3$ at $t=2$, i.e. the last
    timestep that is available, then at $t=3$ can match only one edge. Clearly, the optimal algorithm would match it at 
    $t=2$ and be always better than $A$ in this instantiation and at least as good as $A$ in every other instantiation where
    $u_1$ dies at $t=1$.}
    \label{fig:STG-counterexample}
\end{figure}

Next we prove some useful properties for the optimal algorithm.
Fix any $\gcal(s,t)$. Let ${\pr_M(\gcal(s,t), \gcal(s',t+1))}$ denote the probability that at time $t+1$ 
the snapshot $s'$ of \gcal will appear, given that at time $t$ we had the snapshot $s$ and we matched 
$M \subseteq E(s)$.  Then the following lemma holds.

\begin{lemma}
\label{lem:opt-bellman}
For any stochastic arrival-departure model $\gcal(s,t)$ there exists a {\em deterministic} adaptive 
algorithm that is optimal. In addition, it holds that
$$\chi^*(\gcal(s,t)) = \max_{M \subseteq E(s)} \left\{ |M| + \sum_{\gcal(s',t+1)}\pr_{M}(\gcal(s,t), \gcal(s',t+1))\cdot \chi^*(\gcal(s',t+1))
 \right\}.$$
\end{lemma}
\begin{proof}
Let $A \in \acal$ and let $\pr_A(M)$ denote the probability that algorithm
$A$ chooses to match the edges of matching $M$.
Using the notation introduced above, for every $\gcal(s,t)$, we can express 
$\chi_A(\gcal(s,t))$ as follows.  
\begin{align*}
\chi_A(\gcal(s,t)) 
 & = \sum_{M \subseteq E(s)}\pr_A(M) \cdot \left( |M| + \sum_{\gcal(s',t+1)}\pr_{M}\Big(\gcal(s,t), \gcal(s',t+1)\Big)\cdot \chi_A(\gcal(s',t+1))\right).
\end{align*}

Thus, 
\begin{align}
\nonumber
\chi^*&(\gcal(s,t)) = \max_{A \in \acal} \chi_A(\gcal(s,t)) \\
\nonumber
 & = \max_{A \in \acal} \left\{\sum_{M \subseteq E(s)}\pr_A(M) \cdot \left( |M| + \sum_{\gcal(s',t+1)}\pr_{M}(\gcal(s,t), \gcal(s',t+1))\cdot \chi_A(\gcal(s',t+1))\right)
 \right\} \\
 \label{eq:convexity}
 & = \max_{A \in \acal, M \subseteq E(s)} \left\{ |M| + \sum_{\gcal(s',t+1)}\pr_{M}(\gcal(s,t), \gcal(s',t+1))\cdot \chi_A(\gcal(s',t+1))
 \right\}\\
 \nonumber
 & = \max_{M \subseteq E(s)} \left\{ |M| + \sum_{\gcal(s',t+1)}\pr_{M}(\gcal(s,t), \gcal(s',t+1))\cdot \max_{A \in \acal}\chi_A(\gcal(s',t+1))
 \right\}\\
 \label{eq:chis-def}
 & = \max_{M \subseteq E(s)} \left\{ |M| + \sum_{\gcal(s',t+1)}\pr_{M}(\gcal(s,t), \gcal(s',t+1))\cdot \chi^*(\gcal(s',t+1))
 \right\}.
\end{align}

In the calculations above we get Equation~\eqref{eq:convexity} because of convexity, 
and Equation~\eqref{eq:chis-def} due to the definition of $\chi^*$. In addition, we can see 
that there exists an optimal algorithm where it can choose deterministically a matching $M$ 
in Equation~\eqref{eq:convexity}, hence there exists a deterministic optimal algorithm.
\end{proof}

\section{Approximating $\opt(\gcal)$}
\label{sec:fpras}
In this section we present a fully polynomial randomized approximation scheme 
(FPRAS) for computing $\opt(\gcal)$.  This FPRAS will be the backbone of our
optimal algorithm presented in Section~\ref{sec:alg}.

Recall, an FPRAS for a function $f$ is a randomized algorithm $g$ that, given input $Y$, 
gives an output satisfying \[(1-\varepsilon)\cdot f(Y) \leq g(Y) \leq (1+\varepsilon) \cdot f(Y)\]
with probability at least $\frac{3}{4}$ and has running time polynomial in both $|Y|$ 
and $\nicefrac{1}{\varepsilon}$.  The value of $\frac{3}{4}$ may be rather low in 
practice, but it has been shown that the same class of problems has an FPRAS if we 
choose any probability $\nicefrac{1}{2} <p<1 $~\cite{JerrumVV86}. Furthermore, the 
probability that the result is within a factor of $1 \pm \varepsilon$ of the true value 
can be increased from $\frac{3}{4}$ to $1-\delta$ for any positive $\delta$, just by 
taking the median answer from $O(\log{\frac{1}{\delta}} )$ runs of the algorithm~\cite{JerrumVV86}.

In the following theorem, we derive FPRASs for three different objectives for any \gcal: 
the expected size of an optimal matching, denoted $\opt(\gcal)$;
the expected size of optimal matching given that we match at least some edge $e$ present at timestep 1 in \gcal,
denoted $\opt(\gcal | e)$; 
and the expected size of an optimal matching size given that we do {\em not} match any edge at timestep 1, 
denoted $\opt(\gcal | \emptyset)$.

\begin{theorem}
	For any stochastic arrival-departure model \gcal, there is an FPRAS for
	$\opt(\gcal)$, $\opt(\gcal | e)$, and $\opt(\gcal | \emptyset)$.
\end{theorem}
\begin{proof}
We begin by explaining the algorithm for approximating $\opt(\gcal)$, and follow with two 
slight adaptations of it to allow for computing $\opt(\gcal | e)$ and $\opt(\gcal | \emptyset)$.
	
So, let $\gcal := \langle G,a,b,\pcal \rangle$ and let $G=(V,E)$.
Sample an instantiation $I \in \ical_\gcal$ of \gcal and compute a maximum matching $M_I$ for it.
This can be clearly done in polynomial time.
Let $X=|M_I|$. Clearly, $\opt(\gcal) = E[X] = \sum_{I \in {\ical_\gcal}} \pr(I)*\cdot |M_I|$. 
In addition, since $X$ is the size of a matching, it clearly holds that
	\begin{equation}\label{eq:estimator_variance}
		\sigma(X) \leq \lfloor |V|/2 \rfloor +1 .
	\end{equation}
	
We perform the above experiment for $X$ independently $k$ times. Let $X_1, \ldots, X_k$ 
be the respective maximum matching sizes and consider the estimator 
$X(k) = \frac{X_1+\ldots+X_k}{k}$.  Notice that $X(k)$ is unbiased, since $E[X(k)] = E[X]$. 
We also have that $\sigma(X(k)) =\frac{\sigma(X)}{\sqrt{k}}$; see~\cite{vazirani_approx}.
	
	So, for any $\varepsilon>0$, it holds:
	\begin{eqnarray}
		\notag \Pr\Big[ | X(k) - E[X(k)] | \geq \varepsilon \cdot |X(k)| \Big] &\leq& \left( \frac{\sigma\left( X(k) \right)}{\varepsilon E[X(k)]} \right)^2 \\
		 \notag &=& \left( \frac{\sigma\left( X \right)}{\varepsilon \cdot \sqrt{k} \cdot E[X]} \right)^2 .
	\end{eqnarray}
	
	The latter, assuming $E[X] \geq 1$ and for $k=\frac{1}{\varepsilon^2}\cdot n^4$, gives: \[Pr\Big[ | X(k) - E[X(k)] | \geq \varepsilon \cdot |X(k)| \Big ] \leq \left( \frac{\lfloor \frac{n}{2} \rfloor  +1}{\varepsilon \cdot \sqrt{k}} \right)^2 \leq  \frac{1}{n^2}  . \]
	
	It could be the case that $E[X] \leq 1$; indeed, if no edge appears in a realization, then the maximum matching size is $X=0$. The probability $Pr[X=0]$ is bounded above by the probability that no edge appears in a possible maximum matching having edges of the highest possible probability of occurrence. The latter can be easily calculated in any particular given \gcal; if this probability is greater than $1-\frac{1}{n}$ then the algorithm outputs the estimator $X=0$. Otherwise, it outputs $X(k)$.
	
	We now proceed with the adaptation of the algorithm to approximate $\opt(\gcal | e)$.
	Before sampling $X$, we remove $e$ and all adjacent edges from the given graph. In the resulting graph $G'$ we perform the experiment for $X$ as before. The estimator for $\opt(\gcal | e)$ is $1$ plus the estimator for $X$ in $G'$. Notice that $G'$ has $n-2$ vertices so Equation~\ref{eq:estimator_variance} still holds and our analysis follows.
	
	Finally, to approximate $\opt(\gcal | \emptyset)$, we make the following adjustment to the algorithm 
	for $\opt(\gcal)$. For each of the $k$ experiments for $X$, we remove from the produced realization all vertices that are present {\em only} on timestep 1. The estimator for $\opt(\gcal | \emptyset)$ is the estimator for $X$ in the resulting graphs. We get our result by  noticing that Equation~\ref{eq:estimator_variance} holds again, since each instantiation has at most $n$ vertices.
	
\end{proof}
\section{An optimal algorithm for two timesteps}
\label{sec:alg}

In this section we derive a probabilistic, polynomial-time optimal algorithm for 
stochastic arrival-departure models with two timesteps. Our algorithm utilizes 
the FPRAS from the previous section. 

A first attempt to derive an optimal algorithm would be to estimate the value of the optimal 
matching given that we match some edges at timestep 1 and match the set of edges that 
maximize this value. However, even if this approach was correct, we would have to evaluate 
an exponential number of subsets of edges, which is inefficient. On the other hand, we observe 
that we do not have to check all the edges of a matching simultaneously, but we can {\em create}
a matching by adding edges one by one. Hence, we propose the following algorithm to use at Step
\ref{gaa:step-choice} of Algorithm~\ref{alg:gaa}.

\begin{algorithm}[htb]
\caption{Split-Matching Algorithm} \label{alg:split-match}
\begin{algorithmic}[1]
 \REQUIRE{A snapshot $s$ at time $t$ of $\gcal =\langle G=(V,E),a,b, \pcal \rangle$.}
 \ENSURE{A matching $M_t \subseteq E(s)$.}
\STATE{$M_t \leftarrow \emptyset$; $\text{advance} \leftarrow 0$; $s' \leftarrow s$;}
\WHILE{$\text{advance} = 0$}
\STATE{Compute $\opt^*(\gcal(s,t)|uv)$ for every edge $uv \in E(s)$; 
Compute $\opt(\gcal(s,t)|\emptyset)$;}
\IF{$\opt(\gcal(s,t)|\emptyset) > \max_{uv \in E(s)}\opt^*(\gcal(s,t)|uv)$}
\STATE{$\text{advance} \leftarrow 1$}
\ELSE
\STATE{$u^*v^* = \argmax_{uv \in E(s)}\opt^*(\gcal(s,t)|uv)$;}
\STATE{$M_t \leftarrow M_t \cup u^*v^*$;}
\STATE{$s' \rightarrow s' - \{u^*, v^*\}$;}
\ENDIF
\ENDWHILE
\end{algorithmic}
\end{algorithm}

\begin{lemma}
\label{lem:same-match}
For any $\gcal(s,t)$, with high probability, Algorithm~\ref{alg:split-match} will choose a matching $M_t$ such that
$|M_t| + \sum_{\gcal(s',t+1)}\pr_{M_t}(\gcal(s,t), \gcal(s',t+1))\cdot \opt(\gcal(s',t+1))$ will be equal to
$\max_{M \subseteq E(s)} \left\{ |M| + \sum_{\gcal(s',t+1)}\pr_{M}(\gcal(s,t), \gcal(s',t+1))\cdot \opt(\gcal(s',t+1))\right\}$.
\end{lemma}
\begin{proof}
Let $\mcal_t$ be the set of optimal matchings for timestep $t$ that maximizes the expected value of the matching 
for the stochastic arrival-departure model $\gcal(s,t)$. In other words, any matching in $\mcal_t$ is a ``correct''
choice at this timestep. Algorithm~\ref{alg:split-match}, with high probability, will produce a matching in $\mcal_t$ by
choosing one edge at a time. The probability to make a wrong choice is bounded by the maximum size of any matching in 
$\mcal_t$, which is at most $\frac{n}{2}$, times the error of the FPRAS. The lemma follows, since we can choose the accuracy of
the FPRAS.
\end{proof}

With Lemma~\ref{lem:same-match} in hand we can state and prove our main theorem.
\begin{theorem}
In stochastic arrival-departure models with two timesteps, Algorithm~\ref{alg:gaa} using 
Algorithm~\ref{alg:split-match} gives an optimal algorithm with high probability.
\end{theorem}
\begin{proof}
Let $\gcal(s,t)$ be a stochastic arrival-departure models with two timesteps.
Recall, that for an optimal algorithm it holds that 
$$\chi^*(\gcal(s,t)) = \max_{M \subseteq E(s)} \left\{ |M| + \sum_{\gcal(s',t+1)}\pr_{M}(\gcal(s,t), \gcal(s',t+1))\cdot \chi^*(\gcal(s',t+1))\right\}.$$ In addition, at the last timestep $t^*$ it holds that $\chi^*(\gcal(s',t+1)) = \opt(\gcal(s',t+1))$. Hence, when there are only two timesteps, i.e., $t \in \{1, 2\}$ we have that
\begin{align*}
\chi^*(\gcal(s,1)) 
& = \max_{M \subseteq E(1)} \left\{ |M| + \sum_{\gcal(s',2)}\pr_{M}(\gcal(s,1), \gcal(s',2))\cdot \chi^*(\gcal(s',2))\right\} \\
& = \max_{M \subseteq E(1)} \left\{ |M| + \sum_{\gcal(s',2)}\pr_{M}(\gcal(s,1), \gcal(s',2))\cdot \opt(\gcal(s',2)) \right\}\\
& =  \max_{M \subseteq E(1)} \left\{ |M| + \opt(\gcal(s',2)| M)\right\}.
\end{align*}
Observe that, due to Lemma~\ref{lem:same-match}, the last equation is exactly what Algorithm~\ref{alg:split-match} 
chooses to match at the first timestep. Hence, the theorem follows.
\end{proof}

\section{An upper bound on the price of stochasticity}
\label{sec:negative}

In this section we prove the following theorem

\begin{theorem}
The price of stochasticity is at most $\frac{2}{3}$.
\end{theorem} 

We will prove our theorem by creating a specific arrival-departure model and then we
will derive the exact value.

Let us denote by $S_n = (V,E)$ the graph on $2n$ vertices 
$\ell_1, \ell_2, \ldots, \ell_n, u_1, u_2, \ldots, u_n$, where 
$L = \{ \ell_1, \ell_2, \ldots, \ell_n \}$ is a clique, $U  = \{ u_1, u_2, \ldots, u_n \}$ is an
 independent set, and $\ell_iu_i \in E$ for every $i \in [n]$, and there are no other edges in $S_n$.

Let $\scal_n$ be a stochastic arrival-departure model with the underlying arrival-departure graph
$\langle S_n,a,b \rangle$, where $a_{\ell} = 1, b_{\ell} = 2$ for every $\ell \in L$, and 
$a_{u} = 2, b_{u} = 2$ for every $u \in U$.
In other words, all vertices of $L$ arrive at timestep 1, all vertices of $U$ arrive at timestep 2, and 
every vertex in $L$ dies at timestep 1 with probability 1/2.
Notice that any instantiation of $\scal_n$ is uniquely defined by a set $C \subseteq L$ of vertices 
that survive until day 2.
We denote the corresponding instantiation of $\scal_n$ by $I_{C}$. It is easy to see that $\opt(I_{C})$
depends only on the number $k = |C|$ of survived vertices and is equal to $k + \lfloor(n-k)/2\rfloor$.

In order to compute $\phi^*(\scal_n)$ it is convenient to introduce for every $v \in L$ an indicator random 
variable $X_v$ which is equal to 1 if and only if $v$ survives until day 2. Let $Y = \sum_{v \in L} X_v$. 
Then $\opt(I) = Y(I) +  \lfloor (n-Y(I))/2 \rfloor$, and therefore

\begin{equation}\label{eq:optSn}
	\begin{split}
		\opt(\scal_n) &= \ebb_{\scal_n}[\opt] = \ebb_{\scal_n}[Y +  \lfloor (n-Y)/2 \rfloor] 
		= \frac{n}{2} + \ebb_{\gcal}[\lfloor (n-Y)/2 \rfloor] =\\
		&= \frac{n}{2} + \frac{n-n/2}{2} - \frac{1}{4} = \frac{3n - 1}{4}.
	\end{split}
\end{equation}

In the next lemma we show that for any $\epsilon > 0$ no optimal algorithm can achieve 
$2/3 + \epsilon$ of $\opt(\scal_n)$ when $n$ goes to infinity. Hence, our theorem will follow.

\begin{lemma}\label{lem:phi_upper}
	$\lim\limits_{n \to \infty} \phi^*(\scal_n) = 2/3$.
\end{lemma}
\begin{proof}
	Let $A \in \mathcal{A}$ be an optimal deterministic algorithm, which exists by Theorem~\ref{lem:opt-bellman}.
	Since $A$ is deterministic, at day 1 it matches a fixed set $M$ of edges connecting vertices in $L$, and it extends 
	this matching to the maximum one at day 2. Notice that each vertex $\ell_i \in L \setminus V(M)$ that survived until 
	day 2 contributes one edge $\ell_iu_i$ to the final matching. Let $|M| = t$, then 
	$A(I_C) = t + |C \setminus V(M)| = t + \sum_{v \in L \setminus V(M)} X_v(I_C)$, and hence
	\begin{equation}\label{eq:chiStarSn}
		\chi^*(\scal_n) = \chi_{A}(\scal_n) = \ebb_{\scal_n}[A] = \ebb_{\scal_n}\left[t + \sum_{v \in L \setminus V(M)} X_v\right] = 
		t + \frac{n-2t}{2} = \frac{n}{2}.
	\end{equation}
	
	\noindent
	Combining (\ref{eq:chiStarSn}) with (\ref{eq:optSn}) we derive 
	$$
		\phi^*(\scal_n) = \frac{2n}{3n - 1},
	$$
	which tends to $2/3$ as $n$ goes to infinity.
\end{proof}

\section{Discussion}
\label{sec:discussion}
In this paper we studied the maximum cardinality matching  problem in the stochastically arrival-departure model. 
We defined the price of stochasticity and we have proven an upper bound of $\frac{2}{3}$ even in arrival-departure models defined
over two timesteps.
Furthermore, we proved the existence of a deterministic optimal algorithm for the problem and we derived an 
optimal algorithm for the fundamental case where we have two timesteps. Our algorithm is probabilistic and it heavily 
relies on the FPRAS we derived for approximating the expected value of an optimal matching.

Our work leaves open several interesting questions and creates a plethora of other challenging and important questions.
The most obvious open question is to derive a polynomial-time optimal algorithm for more than two timesteps. Is our algorithm
indeed optimal for this case? We conjecture that this is the case. A different route would be to aim for non-optimal 
algorithms that achieve good approximation guarantees. We highlight that any algorithm that greedily matches edges
{\em always} creates a maximal matching and thus it is by default a 0.5 approximation of the optimum. A more technical 
question is whether the computation of the {\em exact} optimal value can be done in polynomial time or if it is $\sharp P$-complete.

In addition to the above mentioned questions, we can study other intriguing objectives for the problem. 
Recall, $\chi_A(\gcal) := \sum_{I \in \ical_\gcal} \pr(I) \cdot A(I)$ for some adaptive algorithm $A$, thus in a sense
the objective is to ``be good on the average''. A different objective would be 
$\rho_A(\gcal) := \sum_{I \in \ical_\gcal} \pr(I) \cdot \frac{A(I)}{\opt(I)}$ which would ``penalize'' the algorithm for 
missing edges that should have been matched.

\end{document}